\documentclass[conference]{IEEEtran}
\usepackage[utf8]{inputenc}
\usepackage[english]{babel}
\usepackage{amssymb,amsmath,amsthm}
\usepackage[inline]{enumitem}
\usepackage{graphicx}
\usepackage{caption}
\usepackage{subcaption}
\usepackage{algorithm}
\usepackage{algpseudocode}
\usepackage{algorithmicx}

\usepackage{xspace}
\usepackage{xcolor}

\newtheorem{theorem}{Theorem}[section]

\newtheorem{lemma}[theorem]{Lemma}

\newcommand{\etal}{{\it et.~al}\xspace}

\newif\ifmark
\newif\ifhidenote
\marktrue
\hidenotetrue

\newif\ifspace
\spacetrue

\usepackage[normalem]{ulem}
\ifmark

\newcommand{\del}[1]{\sout{#1}}
\ifhidenote
\newcommand{\note}[1]{}
\else
\newcommand{\note}[1]{{\sffamily\itshape\bfseries\uline{#1}}}
\fi
\else

\newcommand{\del}[1]{}
\newcommand{\note}[1]{}
\fi

\newcommand{\name}{{AuditShare}\xspace}

\begin{document}

\title{\name: Sensitive Data Sharing with Reliable Leaker Identification}

\author{\IEEEauthorblockN{Zhiyi Zhang}
\IEEEauthorblockA{UCLA\\
zhiyi@cs.ucla.edu}
\and
\IEEEauthorblockN{Yu Guan}
\IEEEauthorblockA{Peking University\\
shanxigy@pku.edu.cn}
\and
\IEEEauthorblockN{Xinyu Ma}
\IEEEauthorblockA{UCLA\\
xinyu.ma@cs.ucla.edu}
\and
\IEEEauthorblockN{Lixia Zhang}
\IEEEauthorblockA{UCLA\\
lixia@cs.ucla.edu}}


\maketitle

\begin{abstract}
	As Personally Identifiable Information (PII) data sharing among multiple parties becomes increasingly common, so does the potential for data leakage.
	As required by new data protection regulations and laws, when PII leakage occurs, one must be able to reliably identify the leaking sources.
	Existing solutions utilize watermark technologies or data object allocation strategies to differentiate the data shared with different parties to identify potential leakers.
	However, these solutions lose their effectiveness under several attack scenarios, e.g., a data sender may leak the data and a receiver may deny the reception of certain shared data.
	Worse yet, multiple receivers might collude and apply a set of operations such as intersection, complement, and union to their received datasets before leaking them, making the task of leaker identification even more difficult.
	
	In this paper, we propose \name, a PII dataset sharing system with reliable leaking source identification.
	Firstly, taking advantage of the intrinsic properties of PII data, \name allocates data objects to individual sharing parties by PII attributes.
	Secondly, \name obliviously transfer data between the sender and each receiver and uses a Merkle Tree as an immutable record of the sharing.
	Thirdly, a knowledge-based identification algorithm is proposed to identify a guilty sender or colluding/non-colluding receivers.
	Through our evaluation, we show that:
	\begin{enumerate*} [label=(\roman*)]
		\item With a modest amount of leaked data, \name can accurately (accuracy$>$99.99\%) and undeniably identify all the guilty parties in different cases;
		
		\item It only takes 0.5 second to share 100,000 data objects in \name, which is practical in real-world deployment.
	\end{enumerate*}
	
\end{abstract}

\section{Introduction}
\label{sec:intro}

As the society becomes increasingly online, sharing Personally Identifiable Information (PII)~\cite{nistpii2010} becomes increasingly needed.
For example, according to the US National Coordinator for Health Information Technology (ONC), more than 60\% hospitals in the US share patient information with health care providers~\cite{hospital-share} to benefit patients, e.g., to avoid duplicate tests.
Another common use case is schools sharing student information with online educational applications for personalized services, according to the U.S. Department of Education~\cite{school-share}.

PII data contains PII attributes, e.g., full name, email address, etc., which can be used to identify a particular person, so the protection of PII data has become a major concern.
For example, according to~\cite{breach-report}, 342 breaches of health information have been reported in the first nine months in 2019, affecting more than 39 million individuals in the US.
Under this background, Europe's General Data Protection Regulation (GDPR)~\cite{gdpr} and California's Consumer Privacy Act (CCPA)~\cite{ccpa} all contain provisions and requirements for PII data protection. If one is identified as a PII leaker, he/she will face serious punishment.

As a result, reliable leaking source identification in PII dataset sharing helps to prevent data breaches and is a key enabler for the enforcement of new regulations.
The problem can be formally described as follows:
a data owner (called \emph{the sender}) shares its PII dataset with a number of agents (called \emph{receivers}), where both sides agree that the shared data should never be disclosed to any third party.
When a PII dataset leakage happens, offenders must be identified.

A number of solutions to the leaker identification problem in dataset sharing have been developed in recent years.
At this time, the basic approach is to differentiate the datasets received by different receivers in the data sharing process.
For example, the sender can allocate different portions of data across receivers~\cite{data-leakage-detection}\cite{modelfordataleakage}\cite{bigraph}\cite{ kumar2014detection}, or watermark the original data~\cite{agrawal2002watermarking}\cite{sion2005rights}\cite{qin2006watermark}\cite{zhang2006relational} before sharing it with each receiver.
In this way, the sender can trace the leaker(s) by matching the data entries or checking the watermark of the leaked data.

However, the traceable differentiation proposed in the existing works \emph{loses its effectiveness} in identifying leakers in several cases, making these solutions less practical.
\begin{enumerate}  [leftmargin=*]
	\item Because of dishonest employees~\cite{amazon-employee}\cite{capital-one-employee} or vulnerable firewalls~\cite{british-airways}, the sender itself could be the leaker.
	In this case, existing solutions do not work because they rely on the sender to create differentiation among receivers.
	Worse yet, it is even possible for the sender to frame a receiver because the sender has the upper hand by virtue of being fully aware of the differentiation.
	
	\item Collusion can help leakers to figure out the differentiation by comparing their received copies, so that leakers can manipulate the differentiation (e.g., data transformation, numeric value randomization) or only leak the overlapping data in order to circumvent leaker identification.
	
	\item The differentiation can be denied if there is a lack of evidence for which data was received by which receiver.
\end{enumerate}

In this paper, we propose \name, a data-sharing system with reliable leaker identification in PII dataset leakage.
Intuitively, \name allocates data objects by data entry's \emph{PII attribute(s)} across the receivers in such a way that no single party is fully aware of the allocation results; at the same time, nobody can deny or lie about which data objects it has received.
More specifically, the design of \name includes three basic components, addressing the issues in the existing solutions.
\begin{itemize}  [leftmargin=*]
\item \textbf{Obliviously Allocating Objects by PII Attributes.}
The sender organizes the raw dataset into data objects by PII attribute(s) and then obliviously allocate data objects among receivers.
With Oblivious Transfer (OT)~\cite{firstOT}, data objects will be partially shared with each receiver in such a way that the sender cannot know which objects have been received and the receiver is not aware of the remaining objects.
To compensate the data loss caused by OT, \name leverages fake but indistinguishable data objects~\cite{modelfordataleakage} to ensure that each receiver can receive all the original data objects eventually.

\item \textbf{Identifying Leakers based on Leakers' Knowledge.}
When a leakage happens, \name can identify\footnote{When the leaker entirely erases PII attributes from the leaked dataset, our proposed knowledge-based algorithm cannot identify the leakers; nevertheless, in this case, the leaked dataset becomes non-PII data according to data protection regulations.} the leaker or a number of colluding leakers by inferring the knowledge of the leaker(s) from the leaked data objects.

\item \textbf{Immutably Recording the Data Sharing.}
During the data sharing, \name immutably records the result of the data sharing, preventing any party from lying or denying what data objects have been transferred/received.

\end{itemize}

We implement a prototype system of \name based on 1-out-of-2 OT extension~\cite{asharov2013more} and evaluate it with real-world datasets.
It only takes about 0.5 seconds to transfer 100,000 data objects (each object is about 90 Bytes) to a receiver, showing \name's practicality.
By checking a modest number of leaked data objects, \name obtains $>$ 99.99\% identification accuracy in leakages where there exists
\begin{enumerate*} [label=(\roman*)]
	\item a guilty sender,
	\item a non-colluding guilty receiver,
	or
	\item a number of colluding receivers.
\end{enumerate*}

\vspace{3mm}
\noindent\textbf{Contributions} 
The contributions of our work are threefold.
\begin{enumerate}  [leftmargin=*]
	\item We propose a OT-based sharing system for PII data and a way to immutably record sharing results.
	To the best of our knowledge, \name is the first allocation based sharing system that is capable of identifying a guilty sender.
	
	\item We propose a knowledge-based leaker identification algorithm that is able to identify colluding leakers.
	
	\item We implemented a prototype of \name and evaluated its efficiency and effectiveness, showing that \name can accurately and undeniably detect all leakers.
\end{enumerate}

\vspace{3mm}
\noindent\textbf{Roadmap}
In the rest of the paper, we use an example scenario where a school shares student information with several online educational services to show how \name works.
We introduce the system model in Section~\ref{sec:model} and provide an overview of \name in Section~\ref{sec:overview}.
After that, we introduce the design of each component in \name in Section~\ref{sec:preprocessing}, \ref{sec:sharing}, and \ref{sec:detection}.
We evaluate our work in Section~\ref{sec:evaluation} and have a discussion in Section~\ref{sec:discussion}.
We introduce the related work in Section~\ref{sec:related} and conclude our work in Section~\ref{sec:conclusion}.
\section{System Model}
\label{sec:model}

In this work, we consider a real world example described as followed to illustrate \name.
A school $S$ as the sender wants to share its student information with $N$ different online educational services $R_1, R_2, ... R_N$ as receivers so that students can enjoy personalized services based on their own situation.
The dataset being shared is a structured dataset that contains PII, such as student ID, student name, etc.
When student information from $S$ is found unexpectedly (e.g., found in public Internet), the goal is to identify the party/parties responsible for the leakage.

\subsection{Threats Model}
In this paper, we consider both the unintentional and intentional data leakage.
So, based on the data sharing scenario, we define the threats model as follows:

\begin{itemize} [leftmargin=*]
	\item\emph{A Guilty Sender.}
	The data sender may accidentally disclose its information, e.g. due to dishonest employees~\cite{amazon-employee} or a vulnerable campus firewall. 
	It is also possible that, when a leakage occurs, the sender may frame a receiver by taking advantage of its complete knowledge about the data sharing.
	
	\item\emph{Guilty receivers with/without collusion.}
	\label{sec:model:threat:collusion}
	We assume that a receiver may leak part or all of its received data, and there may be more than one guilty but non-colluding receivers at the same time.
	Moreover, multiple parties may collude to gain advantages over the leakage identification system and escape the responsibility.
	Specifically, colluding parties can share their received dataset and disclose a combination of their datasets, where we assume the combination includes:
	\begin{enumerate*} [label=(\roman*)]
		\item taking the union, intersection, or complement of any number of available datasets
		and
		\item randomly picking a number of data objects.
	\end{enumerate*}
	We consider the cases when the sender also participates in the collusion (Section~\ref{sec:discussion:sender}).
	However, we assume the sender does not collude with all the receivers (i.e., no party is honest).

	\item\emph{Data Manipulation.}
	To avoid being caught by watermark-based identification approaches, malicious receivers can manipulate their received dataset, e.g., changing students' numeric attributes like grades, ages, etc.
	However, we assume such manipulation will not erase the PII, otherwise such leakage may not be considered a crime by laws.
	
	\item\emph{Repudiation.}
	After data leakage, guilty parties may deny the facts of what has been sent and received.
	For example, to escape the responsibility, a receiver who leaked the data can deny the receipt of a specific portion of the data and a guilty sender may lie about what has been sent.
\end{itemize}

Specifically, in respect of the collusive leakage, our proposed system aims to capture each colluding receiver who \emph{contributed} to the leakage.
If a corrupted receiver $R$ does not contribute to a leakage, that is, if other colluding receivers can have sufficient knowledge to make such a leaked dataset even without $R$'s participation, \name cannot identify $R$ as a leaker.

\subsection{Assumptions}
We make the following assumptions when designing \name.

\begin{itemize} [leftmargin=*]
\item \emph{Leaked Data contains Identifying Information.}
We assume the leaked data is PII data, which means the leaked data contains PII attributes (e.g., Social Security Number, full name).
\name cannot identify leaking sources from leaked data with no PII attributes.
In fact, such a leakage is not considered as a PII leakage crime by regulations like GDPR; identifying leaking sources in this scenario is out of this paper's scope.

\item \emph{Trusted Third Parties.}
In \name, we assume an honest-but-curious \emph{Notary} who will certify the sharing process and a fully-trustworthy \emph{Arbitrator} who will run the leaker identification algorithm in the case of a leakage.
In \name, the Notary cannot gain any information about the data being shared.
If there is no leakage, the Arbitrator is not needed; however, when there is a leakage, the Arbitrator will is able to learn all the data objects.

\item \emph{Adequate Number of Leaked Data Objects.}
We assume an adequate number of leaked objects are known to \name's identification algorithm in the event of a leakage.
A very small number leaked objects are not sufficient for \name to identify leakers.

\item \emph{Fake Data's Influence on Receiver Applications.}
We assume these fake but indistinguishable data objects will not affect receiver applications much, and thus \name does not apply to application scenarios that strictly require no fake objects at receiver side.
A more detailed discussion on this assumption is in Section~\ref{sec:discussion:fake}.

\end{itemize}
\section{\name Overview}
\label{sec:overview}

Our proposed system functions by allocating data objects by their PII attributes across receivers with 1-out-of-2 Oblivious Transfer (OT). 
1-out-of-2 OT is a cryptographic protocol in which a sender sends two messages and the receiver will receive one; at the same time, the sender cannot know which message has been received and the receiver is oblivious to the value of the other message.
The use of OT prevents the sender from fully knowing the allocation result, so the sender cannot frame a receiver.
During the sharing, the whole process will be proved immutably by a Merkle-Tree-based credential signed by the notary, preventing any party from denying or lying about the data allocation.
In the event of a data leakage, \name can identify all leakers, including colluding receivers or a guilty sender, with high accuracy by reversely inferring leakers' knowledge from the leaked data.

\name consists of three main steps.

\textbf{1. Data Preprocessing.}
In \name, data preprocessing does two main tasks.
\begin{enumerate*} [label=(\roman*)]
	\item A sender groups original data entries with same PII attributes into \emph{data objects}, each of which is associated to an individual and will be shared as a whole in a later sharing process.
	
	\item Then the sender prepares a number of fake data objects, which are used for potential leaker detection.
	A sufficient number of fake data objects are required to guarantee a high accuracy of leaker identification; the more fake data objects are used, the higher the accuracy can be achieved.
\end{enumerate*}

\textbf{2. Data Sharing.}
After Data Preprocessing, a 1-out-of-2 OT is applied for the sender to share the \emph{data objects}.
Importantly, each receiver will get all the real \emph{data objects} and a random subsect of \emph{fake data objects} while each receiver's collection of \emph{fake data objects} keeps unknown to both the sender and the receiver.
In this process, a trust-but-curious Notary will witness the transfer and publish a certificate as an immutable proof of the sharing result.

\begin{figure}[t]
	\centering
	\includegraphics[width=0.46\textwidth]{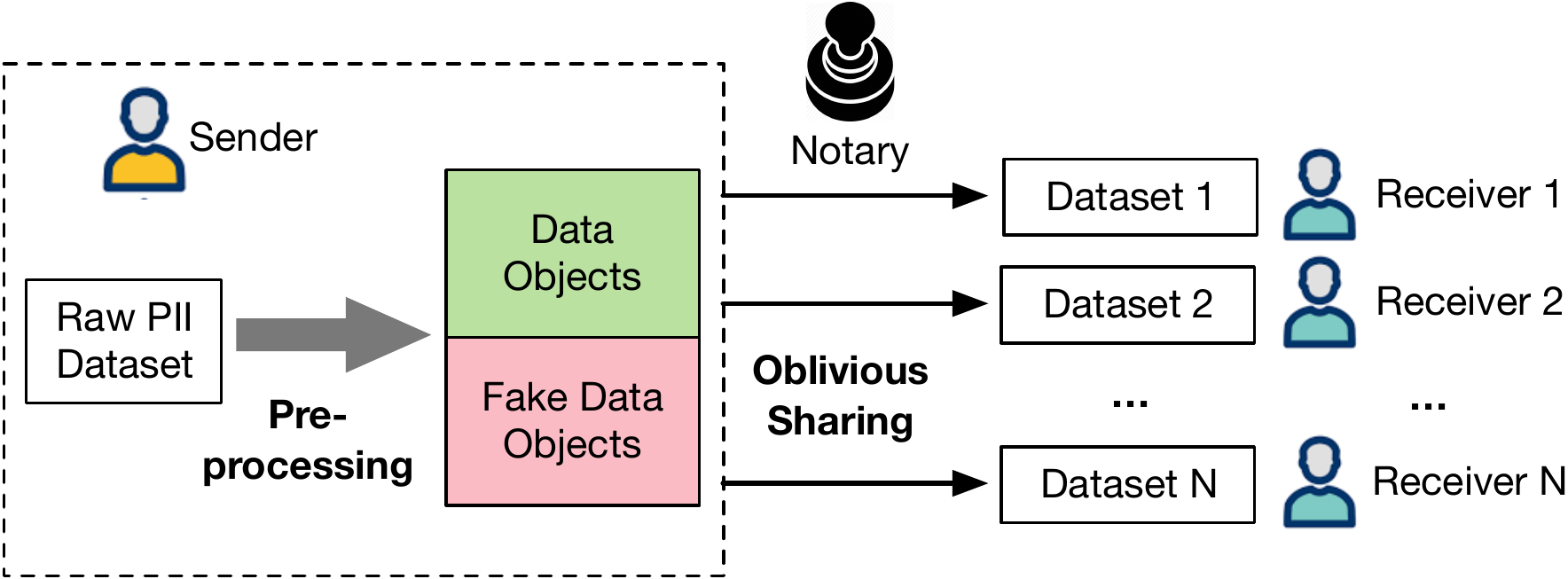}
	\caption{Data Preprocessing and Sharing}
	\label{fig:overview}
	\vspace{-0.5cm}
\end{figure}

\textbf{3. Knowledge-based Leaker Identification.}
In a data leakage, an Arbitrator (e.g., a court, a government agent, or even a smart contract) will trigger the leaker identification process.
All the involved parties will provide the Arbitrator with their received \emph{data objects} and the Notary will disclose the confidential part of sharing certificates to the Arbitrator.
With all the claimed received \emph{data objects} verified, the Arbitrator then triggers the algorithm to identify all the guilty parties.
In detail, the algorithm starts with the hypothesis that a party is not guilty.
Under this hypothesis, the algorithm uses inferential statistics to evaluate the probability of having such a leaked dataset, thus judging whether the party is guilty or not.

Through these three steps, the auditability is built into the sharing system:
By employing 1-out-of-2 OT, each received copy is a different collection from all others; even a combination of several received datasets is also distinguishable from other combinations.
Additionally, no single party is aware of the allocation among receivers.
With the certificate published by the Notary at the end of a sharing process, no party is able to cheat in a later arbitration.
Due to these properties, in the case of a data leakage, our knowledge-based identification algorithm is able to effectively and undeniably identify all leaking sources.
By introducing fake but indistinguishable fake objects, each received dataset contains all the data objects in the sender's original dataset. 
This ensures data receivers are able to provide services to all involved users (e.g., all students of school $S$).

\section{\name: Dataset Preprocessing}
\label{sec:preprocessing}

The purpose of data preprocessing is to prepare data objects for a later sharing process.
Each data object is a collection of data associated to an individual, e.g., each student in school $S$.
In order to allocate different data to receivers, the sender will also generate fake data objects, e.g., data associated with nonexistent students.
Such a process can be illustrated by Figure~\ref{fig:preprocessing}.

\begin{figure}[h]
	\centering
	\includegraphics[width=0.4\textwidth]{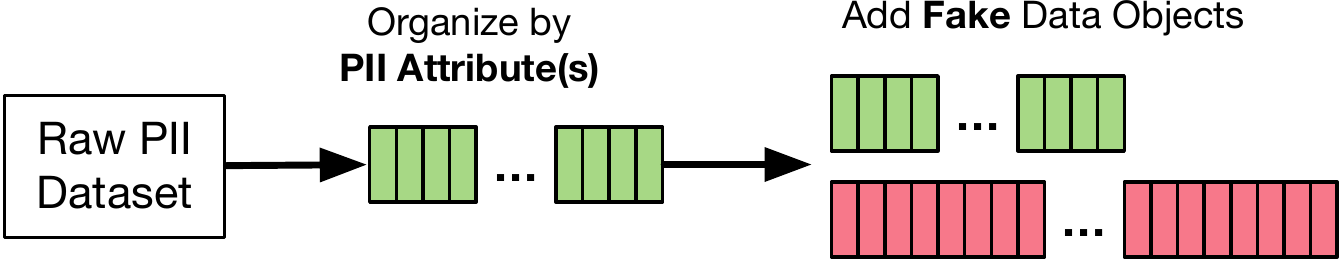}
	\caption{Data Preprocessing}
	\label{fig:preprocessing}
\end{figure}

\subsection{Identifying Data Objects}

A sender first groups raw data entries (e.g., database of student profiles) with the same PII attribute into one \emph{data object}.
Therefore, each object may contain one or more data entries; to combine multiple entries into one object, the sender can simply concatenate multiple entries or perform join operations among tables on PII attributes.
When sharing data with a receiver, a data object is the smallest unit of the allocation; that is, all data entries in each object will be shared as a whole.

To be more specific, the sender will need to select PII attributes as the basis to group the data entries.
Considering a PII dataset, there are two possible scenarios: there exists either a single PII attribute (e.g., student ID) or multiple PII attributes (e.g., student ID and full name).
In the former case, the single attribute will be selected as the basis, while in the second case, a concatenation of all PII attributes will be used (Figure~\ref{fig:attribute-selection}).
In a potential leaker identification process, in case a leaker strips or changes some of PII attributes in leaked data, a leaked object is matched if any of its PII attributes matches the original object.

\begin{figure}[h]
	\centering
	\includegraphics[width=0.45\textwidth]{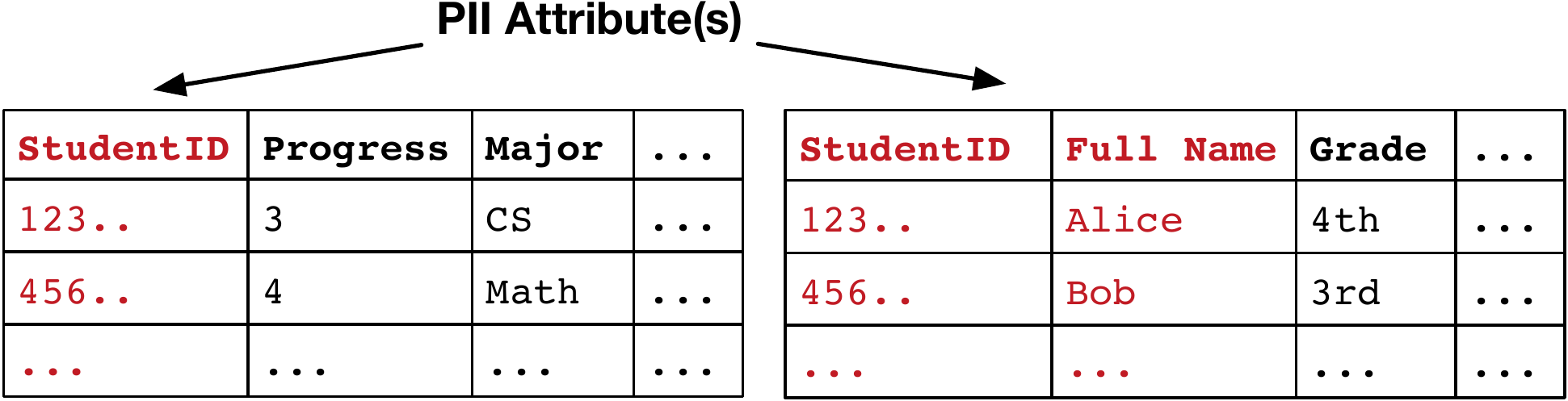}
	\caption{The Selection of PII Attribute(s)}
	\label{fig:attribute-selection}
\end{figure}

\subsection{Preparing Fake Data Objects}
\label{sec:prepare_fake}

In \name, the leaker identification is based on data allocation.
To eliminate data loss caused by the allocation and improve the identification accuracy, the sender will generate fake data objects for the purpose of allocation.
These fake objects do not correspond to real entities but appear realistic to data receivers.
The fake objects can be generated by different means (e.g., differential privacy approaches) and the details are out of the scope of this paper.

Since fake objects are used for the allocation, in general, the more fake objects are generated, the higher the leaker identification accuracy can be achieved.
In \name, we define the lower bound of the quantity of fake objects to be \emph{twice} the size of real data objects.
This is to ensure no real data loss in a later OT as shown in Section~\ref{sec:sharing}.
However, a minimum quantity of fake objects may not ensure a good identification accuracy.
On the other hand, too many fake objects are not needed because the accuracy is high enough and extra objects will lead to a high sharing overhead.

In \name, the proper fake object quantity depends on the total number of receivers and the expectation on identification accuracy.
For example, with 6 receivers and $20,000$ fake objects, when 5\% of the received data is leaked, a 99.99\% accuracy in leaker identification can be guaranteed (Section~\ref{sec:evaluation}).
However, when there are 9 receivers in total, to achieve the same level of accuracy, about $200,000$ fake objects are needed.
Section~\ref{sec:discussion:fake-size} will discuss how to decide the size of fake objects when an expected identification accuracy and the total number of receivers are given.
\section{\name: Data Sharing}
\label{sec:sharing}

After the sender finishes preprocessing data, all original and fake data objects will be shared through 1-out-of-2 Oblivious Transfer (OT), and as a result, each receiver will obtain all real objects and a random part of fake objects while the allocation of fake objects among receivers keeps unknown to the sender, preventing the sender from framing an innocent receiver.

\begin{figure}[h]
	\centering
	\includegraphics[width=0.33\textwidth]{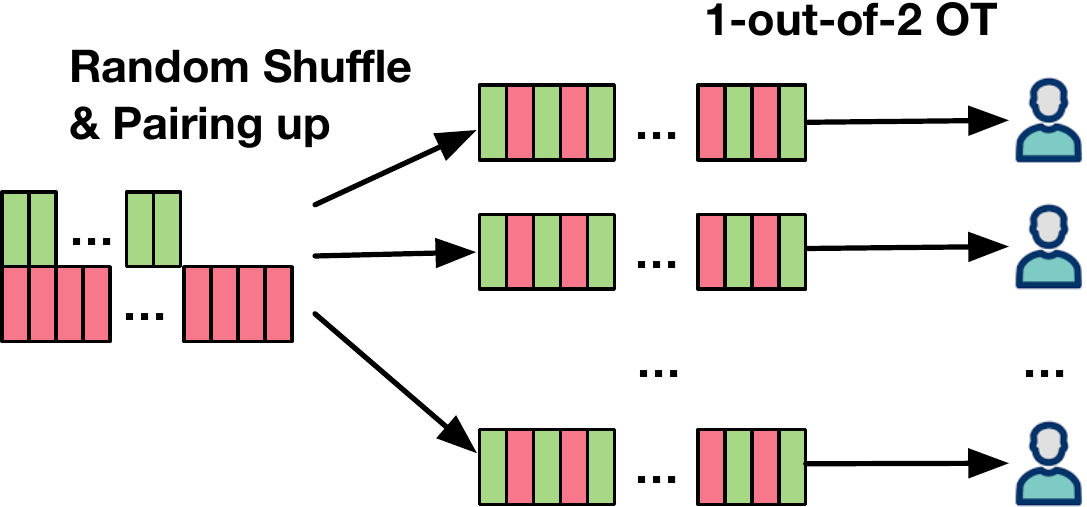}
	\caption{Oblivious Sharing}
	\label{fig:oblivious-sharing}
\end{figure}

\subsection{Oblivious Transfer}

An OT protocol is a cryptographic protocol in which the sender transfers many pieces of information to a receiver but remains oblivious to which piece has been received.
At the same time, the receiver learns nothing but the selected information.
In a classic 1-out-of-2 OT, the sender inputs two messages $m_0$ and $m_1$ while the receiver has a secret bit $b \in \{0, 1\}$.
The receiver will learn $m_b$ while $b$ remains unknown to the sender at the end of the OT.

In this paper, we use the 1-out-of-2 OT.
Since the size of data objects in a PII dataset can be very large, a large number of OTs may be involved in \name.
Therefore, in our prototype implementation and real-world deployment, an OT extension will be used for better efficiency.
As pointed in~\cite{asharov2013more}, the performance bottleneck of oblivious transfer has been in the network rather than the computation overhead.

\subsection{Data Sharing with OT}
In a situation when fake objects are exactly twice the size of real objects, the sender binds each real data object $d_i$ with two different fake objects $f_{i,0}$ and $f_{i,1}$.
After that, the sender creates two tuples $m_0 = <d_i, f_{i,0}>$ and $m_1 = <d_i, f_{i,1}>$ that will be transferred obliviously to the receiver.
As noticed, $m_0$ and $m_1$ have the same real object but contain different fake objects.
As a result of 1-out-of-2 OT, a receiver can only get $m_b$ ($b = 0$ or $1$ represents the receiver's choice).
$m_b$ contains two objects that are indistinguishable to the receiver.
All remaining data objects will be shared in the same way.

When fake data objects are more than twice the size of real objects, extra fake objects should be uniformly distributed into each pair of $m_0$ and $m_1$.

In \name's OT-based sharing, all messages should have the same length in order to prevent side channel attacks.
For this purpose, simple bit padding mechanisms can be used.

\subsection{Provable Sharing}
To make an immutable proof of a data sharing process, \name requires a trusted third party as a \emph{Notary} to generate a credential, so that in a potential arbitration phase, no party can deny the result of the sharing.

In fact, having a notary in commercial data sharing has become a common practice.
In \name, a Notary will generate a certification on what data objects have been transferred without gaining any knowledge of the exact PII information.
To be more specific, the sender, all receivers, and the Notary will obey the following protocol during a data sharing.
\begin{enumerate} [leftmargin=*]
    \item Before a sharing, the sender creates a three-layer Merkle tree, called an \emph{Sending Tree}, which is defined as,
    \begin{center}
        $SendingTree = H(... ||~H(H(m_{i,0})||H(m_{i,1}))~|| ... )$
    \end{center}
	for each message pair $i$ sent by the sender.
    Each two paired OT messages share a parent node and all parent nodes are children of the root node.
    The sender provides the Notary with all nodes in its Sending Tree, which are sufficient to snapshot all the OT messages sent.

    \item Before the sharing, a receiver provides the Notary with choices that will be used in the OT.

    \item After the sharing, the receiver calculates a two-layer Merkle tree, called a \emph{Receiving Tree}, on all the received messages.
    A Receiving tree is defined as,
    \begin{center}
        $RecivingTree = H(... ||~H(m_{i, b})~|| ... )$
    \end{center}
    for each message pair $i$ sent by the sender, where $b$ is receiver's choice in the 1-out-of-2 OT of $i$.
	The receiver then sends the root node value of its \emph{Receiving Tree} to the Notary.

    \item The Notary computes another \emph{Receiving Tree} based on the \emph{Sending Tree} and choices provided by the receiver.
    It then compares the root node with the root node value of the \emph{Receiving tree} provided by the receiver in step 3.
\end{enumerate}
If the comparison in step 4 is successful (i.e., two identical values), the sharing is finished and the Notary will issue a certificate.
Such a certificate contains
\begin{enumerate*} [label=(\roman*)]
	\item an encrypted root node value of the \emph{Sending Tree},
	along with
	\item an encrypted list of choices made by the receiver.
\end{enumerate*}
The Notary will sign the credential and permanently keep it; for example, it can serialize the credential to paper copies or insert it into an immutable ledger.
On the other hand, if the comparison fails, the receiver should delete all the received data with the Notary as a witness and a new sharing process should be started.

In the sharing process, a Notary obtains zero knowledge of the actual data being transferred because of the one-way hash function used in Merkle tree.

\begin{figure*}[t]
	\centering
	\centering
	\begin{subfigure}{.23\textwidth}
		\centering
		\includegraphics[width=0.8\textwidth]{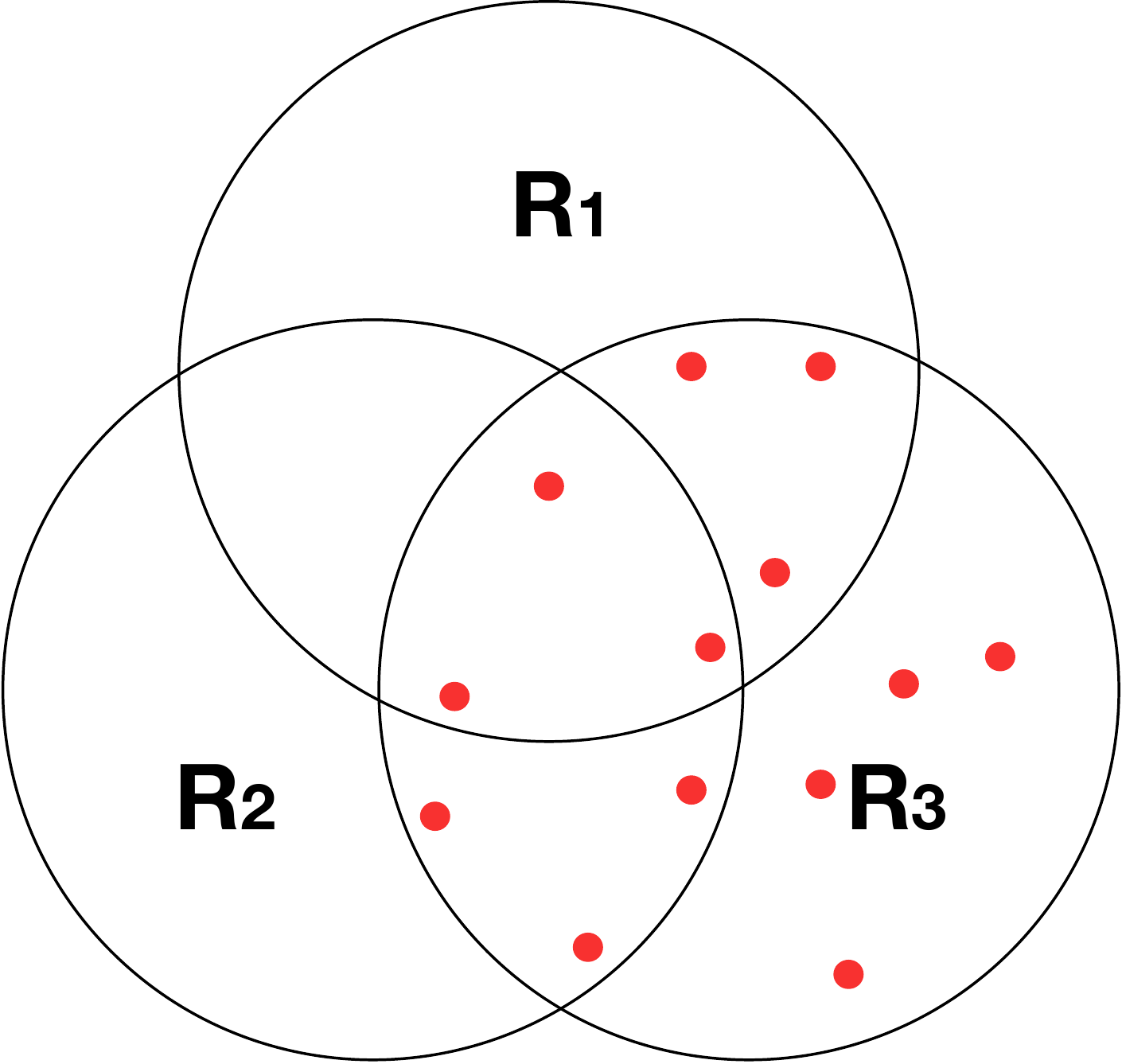}
		\caption{$R_3$ is the leaker}
		\label{fig:arbitration-example-single}
	\end{subfigure} %
	~
	\begin{subfigure}{.23\textwidth}
		\centering
		\includegraphics[width=0.8\linewidth]{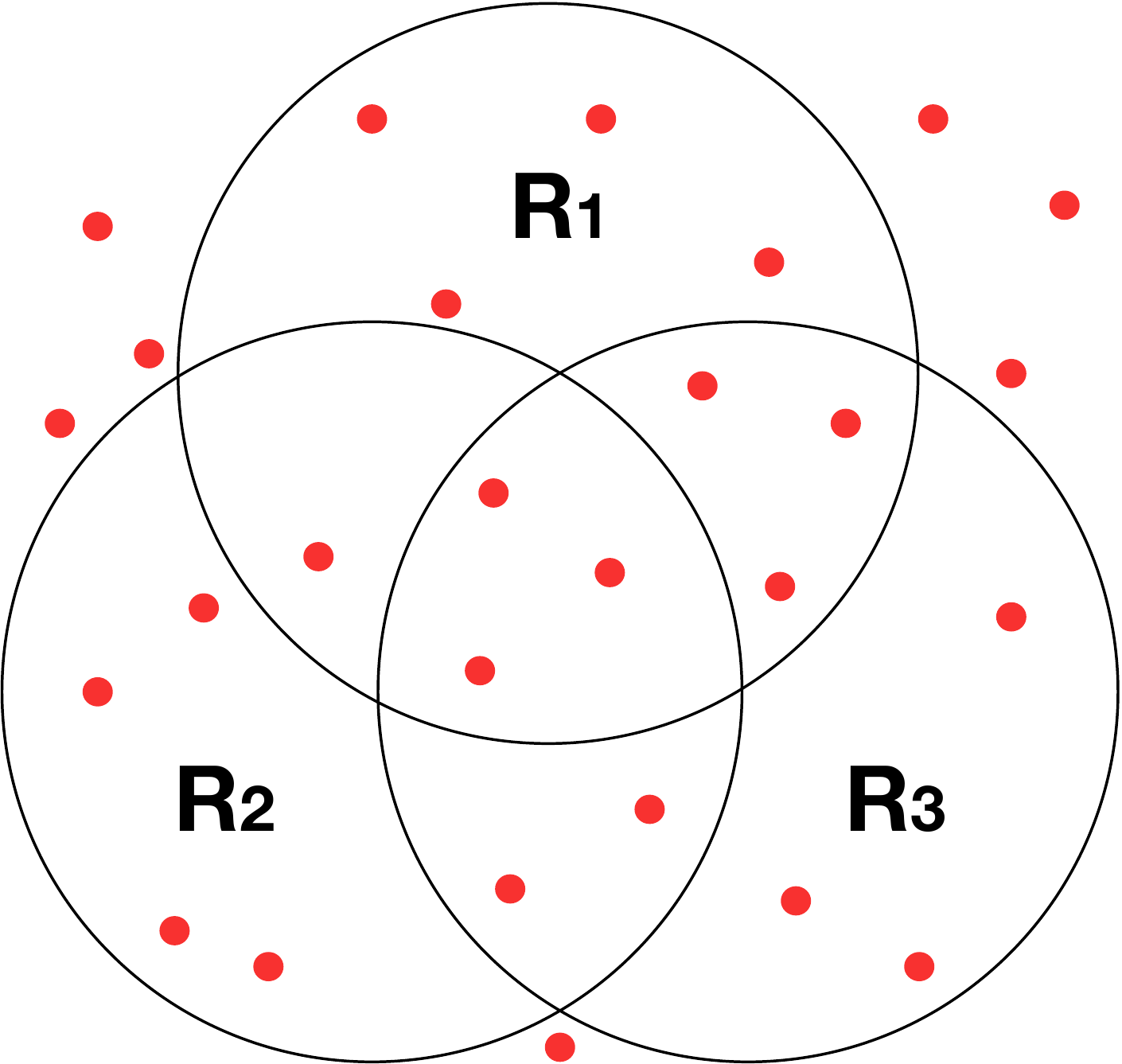}
		\caption{Sender is the leaker}
		\label{fig:arbitration-example-sender}
	\end{subfigure}
	~
	\begin{subfigure}{.23\textwidth}
		\centering
		\includegraphics[width=0.8\linewidth]{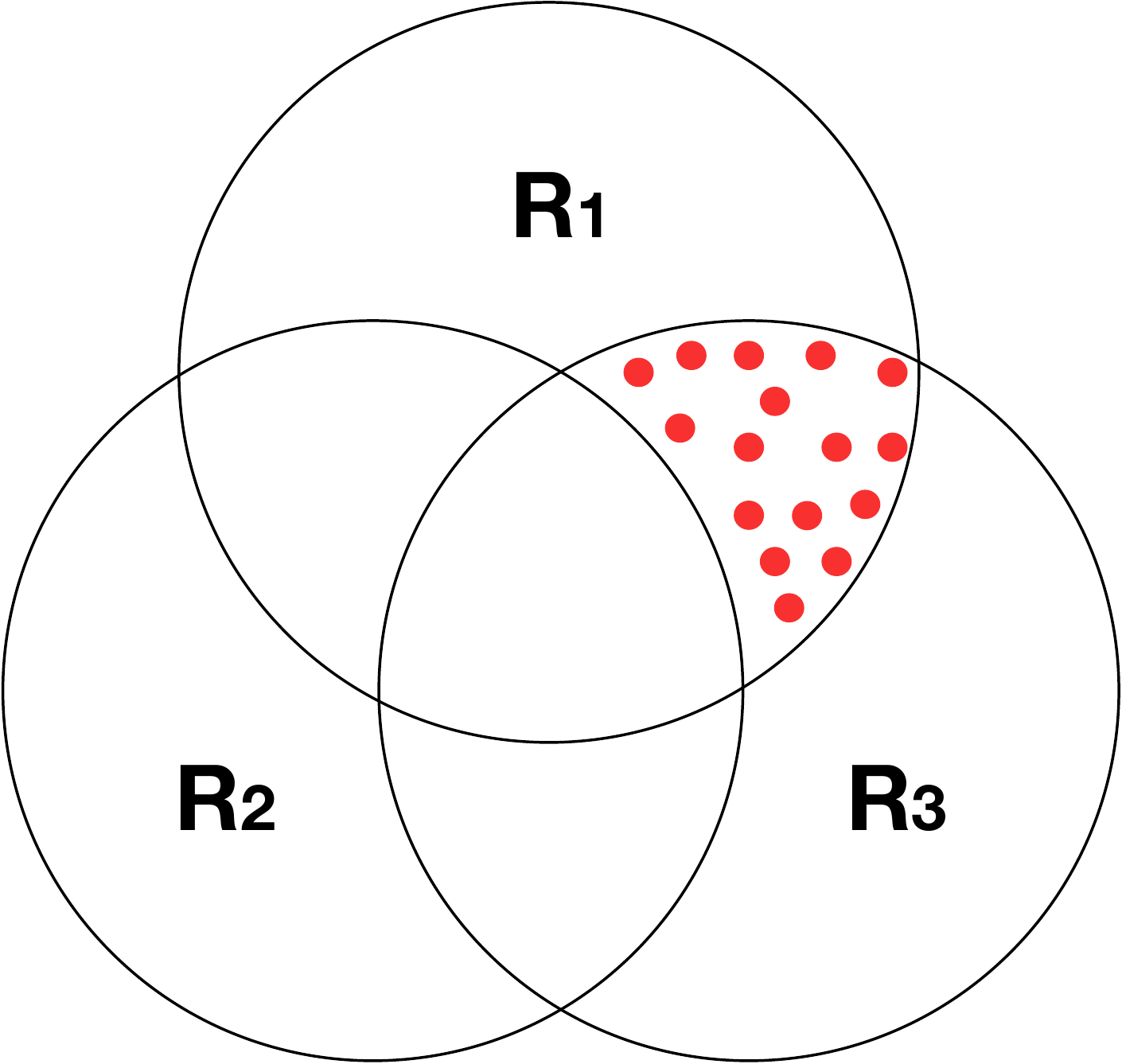}
		\caption{All receivers collude}
		\label{fig:arbitration-example-collude}
	\end{subfigure}
	~
	\begin{subfigure}{.23\textwidth}
		\centering
		\includegraphics[width=0.8\linewidth]{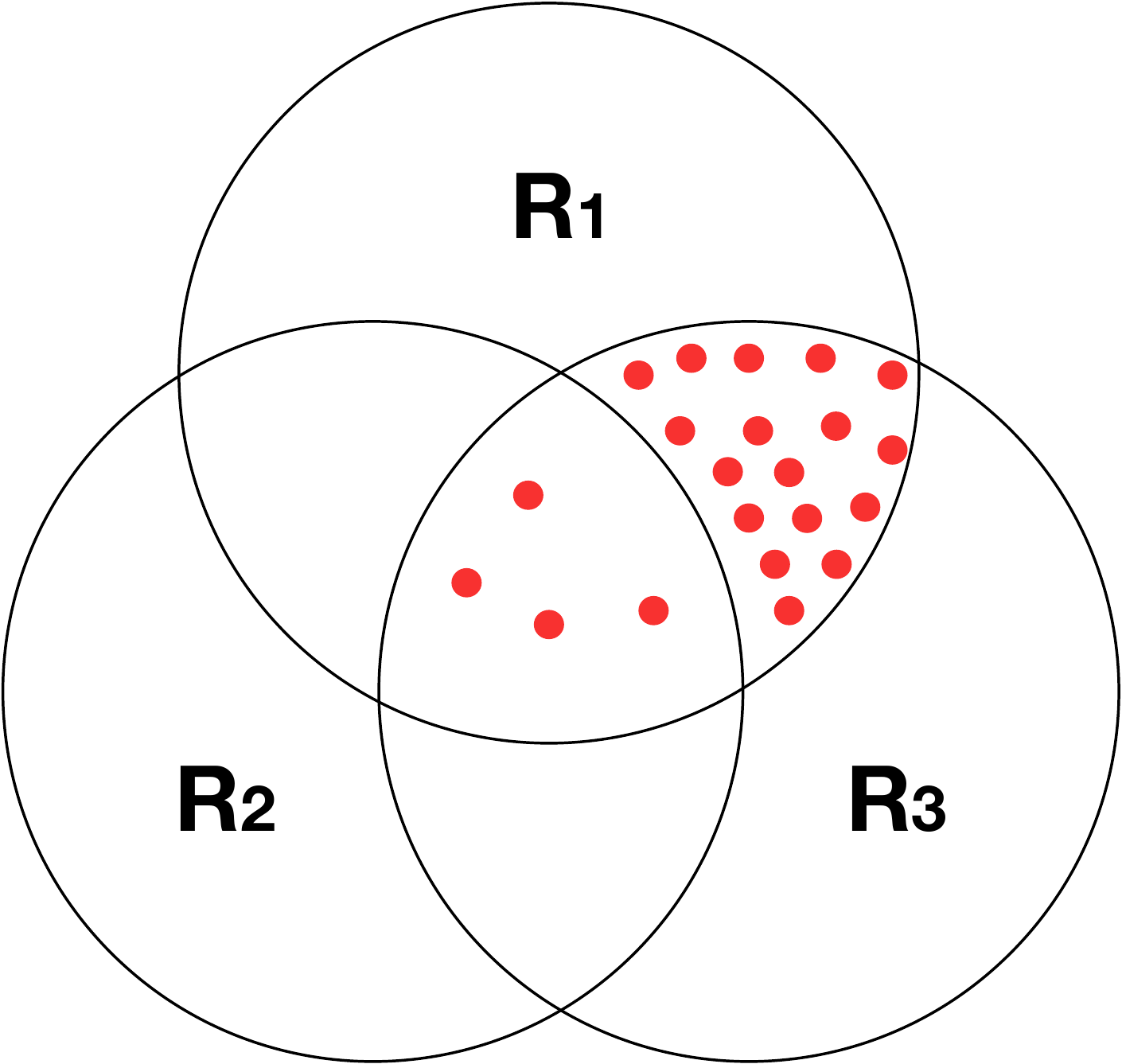}
		\caption{All receivers collude}
		\label{fig:arbitration-example-collude2}
	\end{subfigure}
	\caption{Knowledge-based Leaker Identification (Circles are receivers' knowledge and red points are leaked data objects)}
	\label{fig:arbitration-example}
\end{figure*}

\section{Leaker Identification after Data Leakage}
\label{sec:detection}

\name identifies leakers by analyzing the leaked data and finding the involved parties whose knowledge is sufficient to create such a leaked dataset.
Specifically, since all receivers will have the complete set of real objects, all the objects discussed in this section refer to fake data objects.
The notations used in this section are listed in Table~\ref{tab:notation}.

\begin{table}[t]
	\caption{Notation Table}
	\small
	\renewcommand{\arraystretch}{1.2}
	\label{tab:notation}
	\centering
	\begin{tabular}	{ | p{1.2cm} | p{6.5cm}| } 
		\hline
		$N$ & number of data receivers\\ 
		\hline
		$Data$ & collection of all the fake data objects \\ 
		\hline
		$Data'$ & leaked fake data objects \\ 
		\hline
		$Recv_i$ & fake data objects received by receiver $R_i$, $i > 0$ \\ 
		\hline	
		$Unique_i$ & fake data objects that are only received by $R_i$ \\ 
		\hline	
		$Drop_i$ & fake data objects dropped by $R_i$ in the OT\\ 
		\hline	 
	\end{tabular}
	\vspace{-0.5cm}
\end{table}

To provide some intuitive ideas of how \name's identification works, we show four leakage examples in Figure~\ref{fig:arbitration-example}.
The dots in the figure are the leaked fake data objects while each circle represents all fake data objects received by each receiver.
The intersection of circles are the objects that are received by multiple receivers.
As shown, if leaked data objects are uniformly distributed, as in Figure~\ref{fig:arbitration-example-single}, in receiver $R_3$'s circle, we can easily infer $R_3$ to be the leaker.
When leaked objects are like Figure~\ref{fig:arbitration-example-sender}, the sender is guilty because some leaked objects have never been received by any receiver (dots not falling into any circle).
If objects are distributed in a specific section as shown in Figure~\ref{fig:arbitration-example-collude}, it is trivial to deduce the leaked data is from $Recv_1 \cap Recv_3 \cap Recv_2^{\complement}$, which means all three receivers collude in the leakage.
Figure~\ref{fig:arbitration-example-collude2} is a more complicated scenario where leaked objects are all from $Recv_1 \cap Recv_3$ but with a non-uniform distribution.
Under this scenario, we can still infer that $R_2$ participates in the collusion because $R_1$ and $R_3$ should not have the knowledge to distinguish objects that belong to $R_2$ without $R_2$'s help. 

Following this idea, the leaker identification phase in \name consists of two simple steps.
\begin{enumerate} [leftmargin=*]
	\item \textbf{Confirming Sharing.}
	The involved sender and receivers provide the Arbitrator with their datasets sent/received.
	The Notary also sends decryption keys to the Arbitrator to access the confidential part of certificates.
	In this way, the Arbitrator is able to ensure no party can deny or lie about the result of sharing.
	
	\item \textbf{Identifying Leakers.}
    Once the sharing result is confirmed, the Arbitrator knows the exact allocation of fake data objects among receivers.
	According to the knowledge used to construct the leaked data, \name evaluates the innocence of every involved party and identifies the leaker(s) with a very low error rate.
\end{enumerate}
The first step is to ensure that the Arbitrator can get the correct information of the sender and each receiver's knowledge while the second step is to identify the leakers based on each involved party's knowledge.

The first step is straightforward: the Arbitrator checks the Merkle tree and receiver choices against the datasets provided by the sender and receivers. 
Therefore, in the remaining of this section, we will describe the identification process, that is, given the exact data allocation, how to identify 
\begin{enumerate*} [label=(\roman*)]
	\item a non-colluding leaker (the sender or a receiver),
	and
	\item all colluding leakers.
\end{enumerate*}

\subsection{Identifying Non-colluding Leakers}

\noindent\textbf{Algorithm Summary}
\name takes the following steps to identify each non-colluding leaker, which can be the sender or a receiver.
\begin{enumerate} [leftmargin=*]
	\item The algorithm first figures out each receiver's 
	\begin{enumerate*} [label=(\roman*)]
		\item unique data objects that are never received by other receivers
		and 
		\item the objects that are dropped during the OT.
	\end{enumerate*}
	
	\item For each receiver, the algorithm detects 
	\begin{enumerate*} [label=(\roman*)]
		\item the overlap between leaked objects and unique objects
		and
		\item the overlap between leaked objects and dropped objects.
	\end{enumerate*}
	If the first is nonempty but the second is empty, we consider this receiver is guilty.
	
	\item Finally, to verify the sender's honesty, the algorithm calculates the objects that are only known to the sender (i.e., dropped by all the receivers during OT).
	If there exists an overlap between sender-only objects and leaked objects, we consider the sender guilty too.
\end{enumerate}

\noindent\textbf{Algorithm Rationales}
Based on the data objects allocation, both the sender and each receiver have a number of fake objects that are never received by any other parties, and therefore these objects can serve as the clue of non-colluding leakers (the sender or non-colluding receivers).
We use $Unique_i$ to represent the objects that are only received by a receiver $R_i$ and $SenderOnly$ to refer to those objects that are not received by any receiver.
Moreover, by the property of OT, objects dropped by each receiver $R_i$ during the OT , which are denoted by $Drop_i$, are not known to any party.

Therefore the following statements hold when there is no collusion.
Firstly, $R_i$ is a leaker by a large probability if leaked data objects $Data'$ contain objects that are also in $Unique_i$ but no objects from $Drop_i$.
Specifically, $Unique_i$ can be calculated by
	\begin{align*}
	Unique_i = Recv_i \cap \bigg(\bigcup_{j \neq i} Recv_j \bigg)^{\complement} &&\mbox{(6.1.1)}
	\end{align*}
Secondly, $S$ must be a leaker if $Data'$ contains objects from $SenderOnly$.
Specifically, from the data allocation process, we have:
	\begin{align*}
	SenderOnly = 
	(Drop_1 \cap Drop_2 ... \cap Drop_N) &&\mbox{(6.1.2)}
	\end{align*}
	where $EO_i$ represents the $u$ exclusive objects allocated to $R_i$.

\vspace{3mm}
\noindent\textbf{Pseudo Code and Proof}
Algorithm~\ref{alg:single-leaker} shows the pseudo-code of \name's algorithm of non-colluding leaker identification.

\begin{algorithm}[h]
	\caption{Identifying a Non-colluding Leaker}
	\label{alg:single-leaker}
	\begin{algorithmic}[1]
		\Require \\
		Leaked fake data objects $Data'$ \\
		Each receiver $R_i$'s fake data objects, $Recv_i$ \\
		Fake data objects dropped by $R_i$ in the sharing, $Drop_i$
		\Ensure
		Non-colluding leaker set $L$;
		\State Calculate each receiver's unique objects $Unique_i$ from $Recv_1, Recv_2, ..., Recv_N$
		\For {$i$ from 1 to $N$}
		\If {$Data' \cap Unique_i \neq \varnothing$}
		\If {$Data' \cap Drop_i \neq \varnothing$}
		\State Add $S$ into $L$
		\EndIf
		\EndIf
		\EndFor
		
		\State Calculate objects known to the sender only: $SenderOnly$
		\If {$Data' \cap SenderOnly \neq \varnothing$}
		\State Add $S$ into $L$
		\EndIf
	\end{algorithmic}
\end{algorithm}

We prove the effectiveness of Algorithm~\ref{alg:single-leaker} here.
Specifically, we first prove that in a non-collusive leakage, each leaker's unique objects can be observed by a high probability (Lemma~\ref{lemma:single-receiver} and Lemma~\ref{lemma:sender}) with an adequate number of leaked data objects.
Based on this, we prove the algorithm can successfully identify the non-colluding leakers with a high success rate.

\begin{lemma}
	\label{lemma:single-receiver}
	If there exists a non-colluding guilty receiver $R_i$, given $Data'$ as the leaked fake data objects, $Data'$ will contain objects from $Unique_i$ with a probability $P$ as shown in Equation (6.1.3).
\end{lemma}

\begin{proof}
	\label{proof:lemma:single-receiver}
	In the non-collusive leakage, $R_i$ has no idea that which objects are in $Unique_i$, so we can consider that $Data'$ contains objects randomly from $Unique_i$ and $Recv_i - Unique_i$.
	Therefore, we have the probability of ``at least one object of leaked dataset $Data'$ is in $Unique_i$" to be:
	\begin{align*}
	P~(Data' \cap Unique_i \neq \varnothing) =
	1 - \frac{ \binom{|Recv_i|-|Unique_i|}{|Data'|} }{ \binom{|Recv_i|}{|Data'|} }
	&&\mbox{(6.1.3)}
	\end{align*}
	According to the data allocation process with 1-out-of-2 OT, we have
	\[
		|Unique_i| = (\frac{1}{2})^{N} \times |Data|
	\]
\end{proof}

Importantly, $P$ is large enough even when a guilty receiver $R_i$ only leaks a small portion of data objects from $Recv_i$.
For example, with 3 receivers in total, when $|Data| = 20000$, the probability $P > 1-3\times10^{-13}$ even when $R_i$ leaks 1\% of its received data (i.e., 100 fake data objects).

\begin{lemma}
	\label{lemma:sender}
	If there exists a non-colluding guilty sender $S$, given $Data'$ as the leaked fake data objects, $Data'$ will contain objects from $SenderOnly$ with a probability P as shown in Equation (6.1.4).
\end{lemma}

\begin{proof}
	\label{proof:corollary:sender}
	The proof is the same as Lemma~\ref{lemma:single-receiver}.
	The probability of an nonempty intersection of $Data'$ and $SenderOnly$ is
	\begin{align*}
		P~(Data \cap SenderOnly \neq \varnothing) = \\
		1 - \frac{ \binom{|Data|-|SenderOnly|}{|Data'|} }{ \binom{|Data|}{|Data'|} }
		&&\mbox{(6.1.4)}
	\end{align*}
	According to the property of 1-out-of-2 OT, we have:
	\begin{align*}
		|SenderOnly| =
		(\frac{1}{2})^{N} \times |Data|
	\end{align*}
	Similarly, $P$ is large enough even when $S$ only leaks a small portion of data objects from $Data$. 
\end{proof}

\begin{theorem} \label{theorem:algo1}
Algorithm~\ref{alg:single-leaker} can identify all non-colluding leakers with success rate shown in Equation (6.1.3) and Equation (6.1.4).
\end{theorem}

\begin{proof}
Code line 4-11 is to identify non-colluding guilty receivers and line 12-15 is to identify a guilty sender.
According to Lemma~\ref{lemma:single-receiver} and Lemma~\ref{lemma:sender}, Algorithm~\ref{alg:single-leaker} can detect the sender with success rate shown in Equation (6.1.4) and detect the non-colluding receiver with success rate shown in Equation (6.1.3).
\end{proof}

\subsection{Identifying Colluding Leakers}

\vspace{2mm}
\noindent\textbf{Algorithm Summary}
\name identifies colluding leakers by the following steps.
\begin{enumerate} [leftmargin=*]
	\item The algorithm first represents all the data objects with a vector, in which each index represents a description (called a \emph{pattern}) of each receiver's state of receipt and the value of each element represents the number of objects that fall into this set.
	Similarly, the algorithm transforms all leaked objects into another vector.
	
	\item For each receiver, the algorithm then infers whether this receiver contributed to the collusive leakage by comparing the leaked-object vector with the all-object vector with a statistical test called binomial test.
	Intuitively, the object distribution in the leaked-object vector should not be significantly different from that in the all-object vector if this receiver did not contribute to the leakage.
\end{enumerate}

\vspace{2mm}
\noindent\textbf{Algorithm Rationales}
In a more complicated scenario where multiple receivers collude, the leaked data can still reflect leakers' knowledge because the leakers cannot know the data objects received by honest receivers.
Following this principle, \name is able to identify all leakers with the following procedure.
If a given party $R_i$ did not contribute to the collusion, the leaked objects in and out of $Recv_i$ should be close to statistically identical, due to the 1-out-of-2 OT used in the sharing process. 
If they are not, then $R_i$ is likely a colluding leaker.

To facilitate the explanation of the algorithm, we first introduce the Pattern Matrix $Pat$, where each row represents a type of allocation among receivers.
Given a number of receivers $N$, $Pat$ is a $2^N \times N$ matrix whose value is:
\[
Pat
=
	\begin{bmatrix}
		0 & 0 & ... & 0 & 0 \\
		1 & 0 & ... & 0 & 0 \\
		0 & 1 & ... & 0 & 0 \\
		&  &... \\
		0 & 1 & ... & 1 & 1 \\
		1 & 1 & ... & 1 & 1
	\end{bmatrix}
\]
The value of each row is a reverse of row index's binary value.
In each row $j$, we define each element $Pat[j][i] = 1$ to represent that $R_i$ has received the set of data objects and $Pat[j][i] = 0$ if this object has been dropped by $R_i$ in OT.
For example, we have:
\[
Pat[0]~\rightarrow~Drop_1 \cap Drop_2 \cap ... \cap Drop_N 
\]
which represents the pattern of ``dropped by all receivers".
\[
Pat[2^N]~\rightarrow~Recv_1 \cap Recv_2 \cap ... \cap Recv_N
\]
which represents the pattern of ``received by all receivers".
In the Pattern Matrix, we call two row indexes (i.e., patterns) $k$ and $k'$ are \emph{paired} on $R_i$ if they are only different at $i$'s slot.
For example, $Pat[0]$ and $Pat[1]$ are paired rows on $R_0$.

Since each leaked fake data object can be categorized in to a pattern, which is a row index in the Pattern Matrix, all leaked fake objects can be represented by a \emph{Dataset Vector}, denoted by $V$.
\[
V = [a_0,~a_1,~...,~a_i,~...,~a_{2^N}]
\]
where each element $V[i]$ represents the number of data objects that fall into the $i$th pattern in the Pattern Matrix.
For example, $V[0]$ is the number of objects that are never received by any receiver and $V[1]$ is the number of objects that are only received by $R_0$.
Similarly, in a Dataset Vector, we call two elements $V[k]$ and $V[k']$ are paired if they are indexed by two paired patterns $k$ and $k'$.

Importantly, for two paired indexes (i.e., patterns) $k$ and $k'$ on receiver $R_i$, if $R_i$ makes no contribution to the leakage, a leaked data object of pattern $k$ or $k'$ should have:
\[
P(d~is~of~pattern~k) = 0.5
\]
This is because leakers have no knowledge to distinguish whether the object $d$ has been received by $R_i$ or not.

In order to measure whether the difference is significant enough, we employ a statistic test, binomial test~\cite{rosen2012discrete}, to obtain a p-value~\cite{wiki:p-value}.
In \name, a binomial test can be represented as follow.
\[
p\_value = Binomial\_Test(V[k], V[k'], 0.5)
\]
To be specific, p-value indicates the probability of having a distribution equal to or more extreme than ($V[k]$, $V[k']$), under the hypothesis that $R_i$ is honest, i.e., $P(d \in V[k]) = 0.5$.
For example, when ($V[k], V[k']$) = (14, 6), the p-value is 0.115, meaning 11.5\% of the time one would expect to have a distribution equal to (14, 6) or more extreme (e.g., (15, 5)) if $R_i$ is honest.
When p-value is less than a pre-defined value $\alpha$ (e.g., $\alpha = 0.05$), it is considered that the hypothesis does not hold and thus $R_i$ is one of the leakers.

A noteworthy fact is that \name cannot identify a corrupted receiver when this receiver makes no contribution to the leakage, that is, other corrupted parties can still have sufficient knowledge to make such a leaked dataset even without this receiver's participation.

\vspace{2mm}
\noindent\textbf{Pseudo Code and Proof}
Algorithm~\ref{alg:colluded-leaker}  shows the pseudo code of the \name's identification algorithm for colluding leakers.

\begin{algorithm}[h]
	\caption{Leakers identifying algorithm.}
	\label{alg:colluded-leaker}
	\begin{algorithmic}[1]
		\Require \\
		Leaked fake data objects\\
		The allocation of all the fake data objects \\
		A significance level: $\alpha$, $\alpha$ = 0.05 by default 
		\Ensure
		Leakers set $L$ with an accuracy value $A$
		
		\State Transform leaked data into a Dataset Vector $V$
		
		\State Initialize a vector $PValues[N]$ = \{1,1,...,1\}
		\For {$i$ from 1 to $N$}
		\For {$k$ from 1 to $2^N$}
		\State Find $k'$ which is paired with $k$
		\State $p\_value_i$ = Binomial\_Test($V[k]$, $V[k']$, $P = 0.5$)
		\If {$p\_value_i$ $<$ $\alpha$}
		\State Add $R_i$ into $L$
		\State $PValues[i]$ = min($PValues[i]$, $p\_value_i$)
		\EndIf
		\EndFor
		\EndFor
		\State Calculate accuracy as $A = \prod_{i \in L}(1 - PValues[i])$
		
	\end{algorithmic}
\end{algorithm}

Now we prove the effectiveness of Algorithm~\ref{alg:colluded-leaker} with a brief proof.
Specifically, we first prove that when a receiver contributed to the leakage, its contribution can be reflected by the object distribution of two paired elements in Dataset Vector (Lemma~\ref{lemma:distribution}).
After that, we can prove the correctness of Algorithm~\ref{alg:colluded-leaker} (Theorem~\ref{theorem: algo-2}).

Let $V$ be the Data Vector for the leaked dataset.
\begin{lemma}
	\label{lemma:distribution}
	If there exist any paired patterns $k$ and $k'$ on $R_i$ where $V[k]$ has a significant difference from $V[k']$, we draw the conclusion that receiver $R_i$ contributed to the leakage with an error rate of the p-value obtained in the binomial test.
\end{lemma}

\begin{proof}
	We give a proof by contradiction.
	Assuming $R_i$ does not contribute to the leakage, since the corrupted receivers cannot know $Recv_i$, their leaked data randomly contains both objects received and not received by $R_i$.
	Therefore, $V[k]$ and $V[k']$ are expected to be the same because corrupted receivers cannot tell objects of pattern $k$ and $k'$ apart.
	This contradicts the statement that $V[k]$ has a significant difference from $V[k']$, so receiver $R_i$ is corrupted and has contributed to the leakage.
	
	The error rate is the probability that a significant difference exists when $R_i$ is honest, which is represented by the p-value obtained from the binomial test.
\end{proof}

\begin{theorem}
	\label{theorem: algo-2}
	Algorithm~\ref{alg:colluded-leaker} can identify all the colluding receivers with the accuracy obtained in code line 18.
\end{theorem}

\begin{proof}
	In Algorithm~\ref{alg:colluded-leaker}, line 6 to line 15 check for each receiver $R_i$, whether there exists paired $V[k]$ and $V[k']'$ such that $V[k]$ is significantly different from $V[k']'$.
	According to Lemma~\ref{lemma:distribution}, each corrupted receiver who contributed to the leakage can be identified with an error rate to be the p-value.
	Therefore, for each identified corrupted leaker $j$, the probability of giving a correct result is $1 - PValues[j]$, where $PValues[j]$ is the smallest p-value among all the p-values for each $V[k]$, $V[k']'$ pair in which a significant difference is observed.
	The accuracy for the whole algorithm is the probability that the algorithm makes a right identification on each captured party, so the accuracy should be
	\begin{align*}
	\prod_{j \in L}(1 - PValues[j])
	\end{align*}
	where L is the set of all leakers captured by Algorithm~\ref{alg:colluded-leaker}. 
\end{proof}
\section{Evaluation}
\label{sec:evaluation}

In this section, we evaluate \name's performance in the data sharing phase and the effectiveness of the knowledge-based identification algorithm when leakage happens.
To be more specific:
\begin{itemize} [leftmargin=*]
    \item We measure the time overhead in \name's data sharing sharing process in terms of the number of data objects and receivers.
    The goal is to show how the cryptographic operations of OT will affect the performance.
    From the evaluation, we show \name is able to finish heavy-load PII database sharing tasks in a short time. 
    
	\item We evaluate \name's effectiveness in identifying leakers in different scenarios: a non-colluding guilty receiver, a guilty sender, and a group of colluding parties.
	Our evaluation shows that with total 20,000 objects and 6 receivers, \name is able to identify all the leakers with accuracy more than 99.99\% in any scenarios if more than 5\% of the objects are leaked and used in the identification algorithm.
\end{itemize}

\subsection{Sharing Performance}

We first evaluate \name's time overhead in OT-based data sharing process.
Since \name requires OT operations in the transfer of each data object, directly applying raw OT can be highly time-consuming.
However, the time complexity caused by heavy-load OT operations can be greatly reduced by utilizing OT extensions~\cite{ishai2003extending, kiraz2007efficient, asharov2013more}.
We implement a prototype of \name with the OT extension proposed by Ishai \etal~\cite{ishai2003extending} and all the evaluation results were obtained from this prototype implementation. 

\vspace{3mm}
\noindent\textbf{Evaluation Settings}
We evaluate \name's performance using a PII database extended from a public dataset called Heart Disease Data~\cite{uci-heart}.
The original data contains patients' SSNs and full names, but these PII attributes have been stripped in the publicly available version for privacy concerns.
For the purpose of evaluation, we extend the size of the dataset and add back these two PII attributes with randomly generated values.
The size of each data object used in the evaluation is 88 Bytes.
The evaluation platform is a 4-core 2.2 GHz Intel Core i7 and the communication channel is a UDP connection between local host ports (so the time consumption is mainly from cryptographic operations).

\begin{figure}[h]
	\centering
	\includegraphics[width=0.35\textwidth]{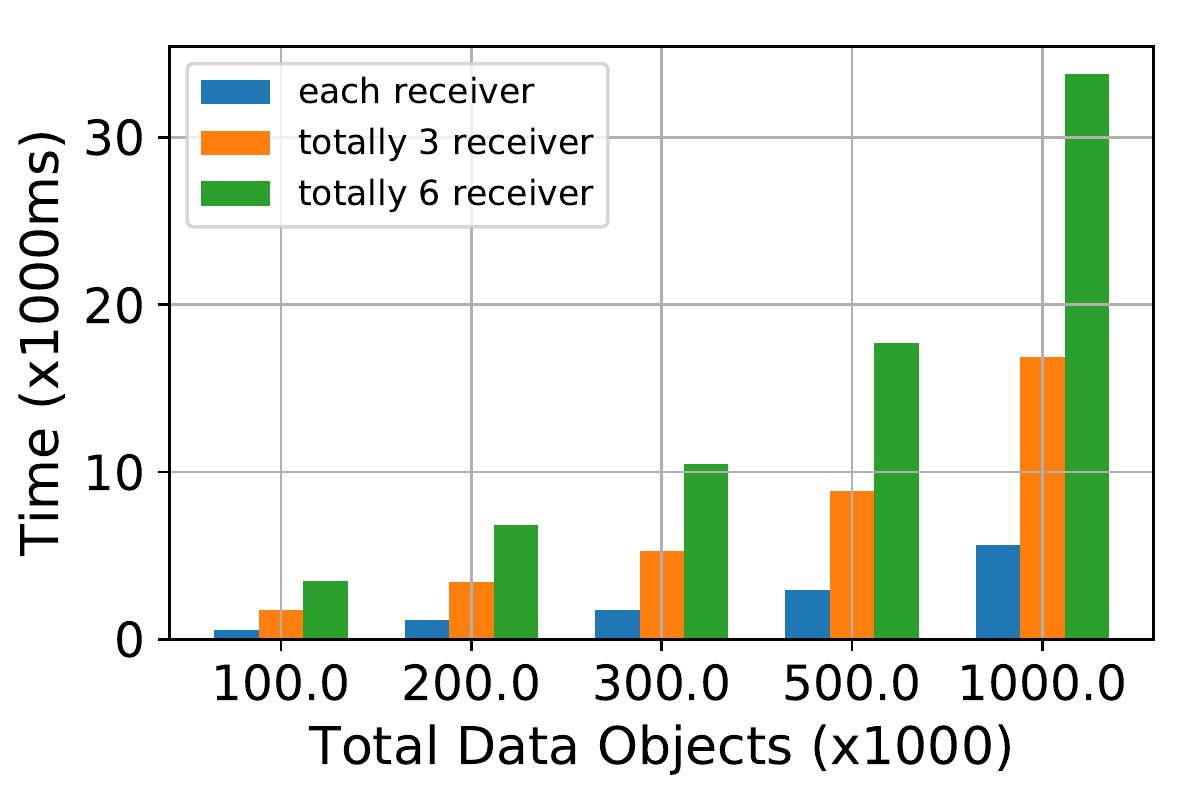}
	\caption{Time Consumption of Sharing in \name}
	\label{fig:performance}
\end{figure}

\begin{figure*} [ht]
	\centering
	\begin{subfigure}{.32\textwidth}
		\centering
		\includegraphics[width=\linewidth]{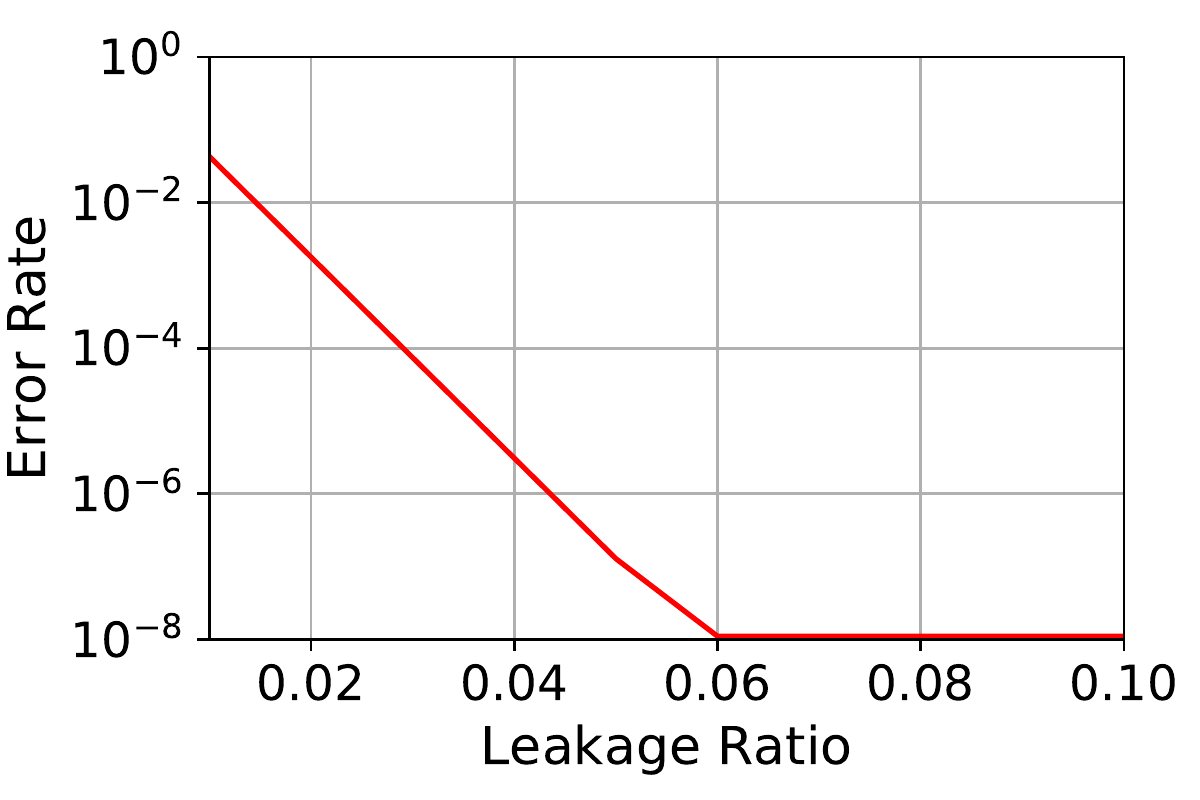}
		\caption{Each Non-colluding Guilty Receiver}
		\label{fig:single-receiver}
	\end{subfigure} %
	~ 
	\begin{subfigure}{.32\textwidth}
		\centering
		\includegraphics[width=\linewidth]{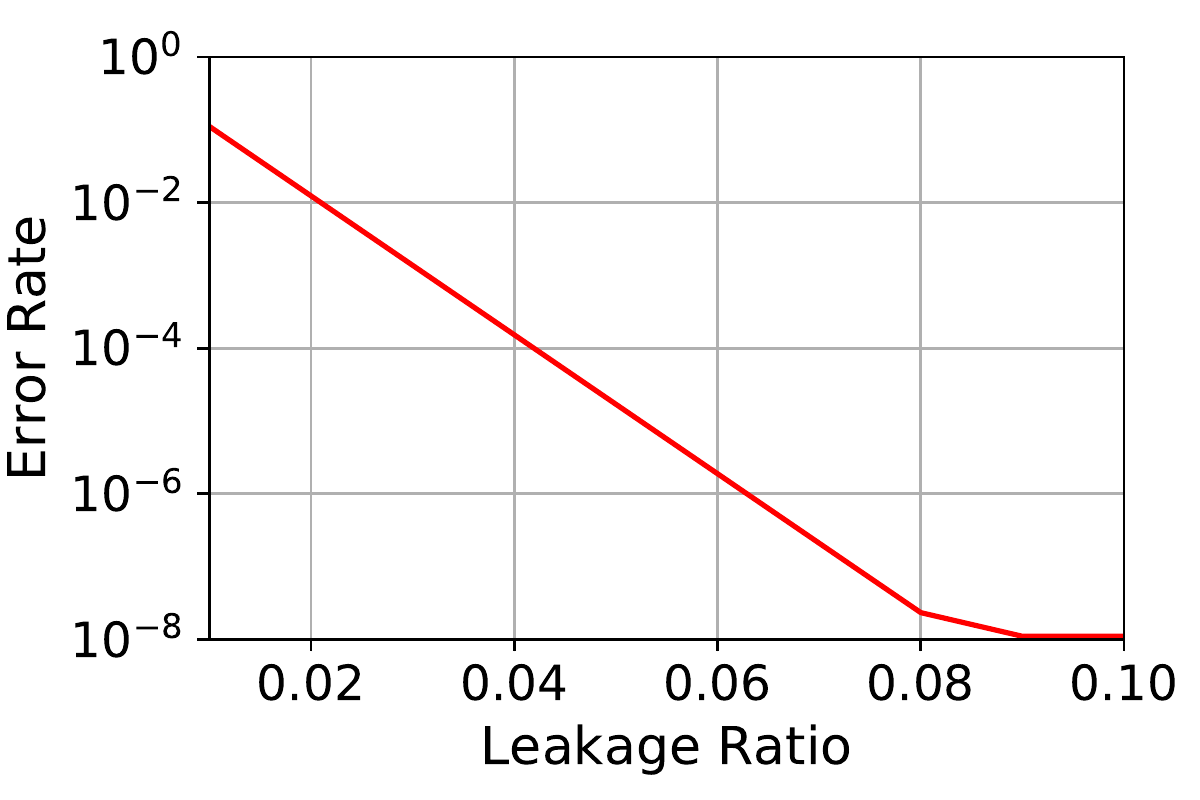}
		\caption{Guilty Sender}
		\label{fig:sender}
	\end{subfigure}
	~ 
	\begin{subfigure}{.32\textwidth}
		\centering
		\includegraphics[width=\linewidth]{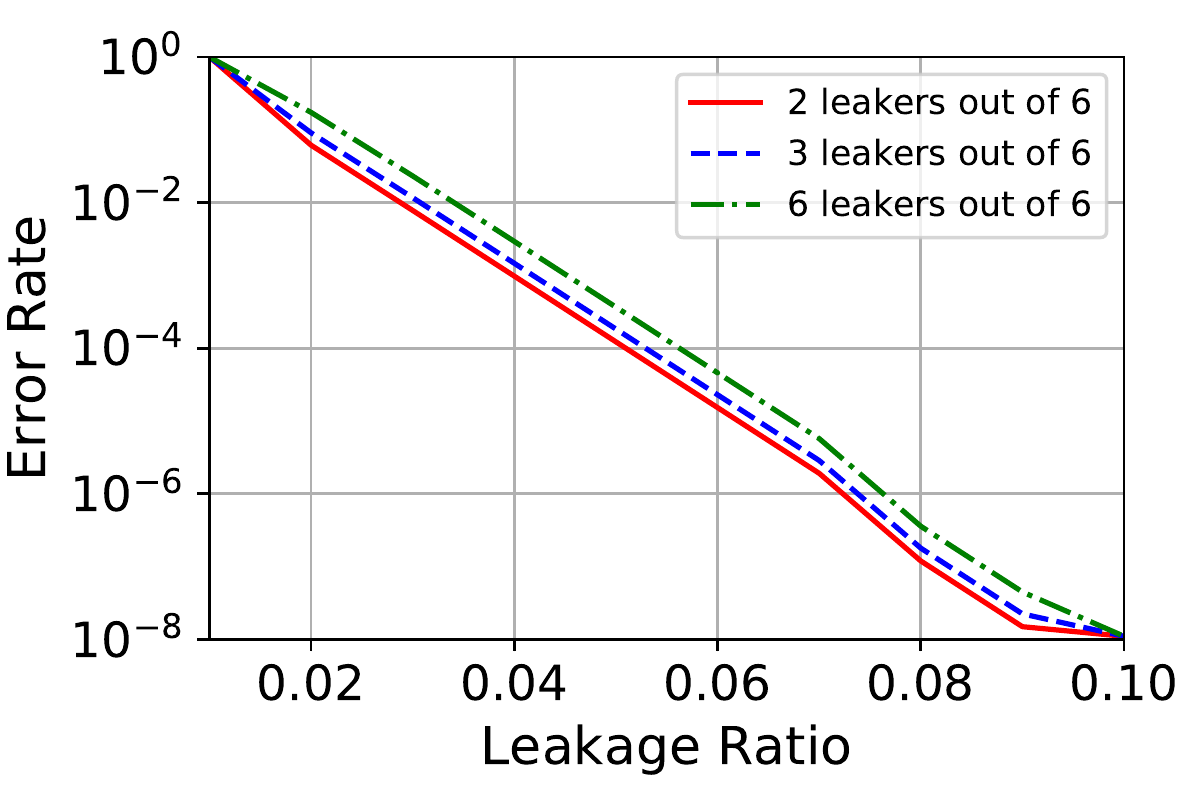}
		\caption{Colluding Receivers}
		\label{fig:collusion}
	\end{subfigure}
	\caption{Leaked Data Amount v.s. Detection Error Rate}
	\label{fig:errorrate-amount}
\end{figure*}

\vspace{3mm}
\noindent\textbf{Evaluation Results}
Our evaluation on \name prototype implementation shows its efficiency in a real-world deployment.
Though the time consumption is linear to the number of receivers and the total objects to share, it only takes about 580 milliseconds to share 100,000 data objects with each receiver.
As shown in Figure~\ref{fig:performance}, even with 1,000,000 data entries and 6 receivers, \name is able to finish data sharing processes in about 30 seconds.
We omitted the time consumption on the preparation and verification in this paper because they take much less time than the sharing process and do not need to be performed in real-time.

\subsection{Accurate Leaker Identification}
\label{sec:evaluation:identification}

We then evaluate the identification accuracy of \name's knowledge-based identification algorithm through data leakages simulations.
To be specific, we allow the leaker(s) to be 
\begin{enumerate*} [label=(\roman*)]
	\item a guilty receiver,
	\item multiple guilty receivers without collusion,
	\item a guilty sender,
	and
	\item multiple guilty receivers with collusion,
\end{enumerate*}
respectively, and simulate the leaked datasets by randomly picking up a certain percentage of objects from the leakers.
In fact, we find the identification accuracy of a non-colluding guilty receivers is independent from how many non-colluding guilty receivers there are (i.e., case (i) and case (ii)), so we merge them into a single case.

In the collusive leakage, we randomly select three combinations; for example, in a leakage simulation of 3 colluding leakers out of 6 receivers, we create the leaked dataset to be a random subset of following sets respectively.
\[
	Recv_1 \cap Recv_2 \cap Recv_3,~Recv_1 \cap Recv_2^{\complement} \cap Recv_3 
\]\[
	(Recv_1 \cap Recv_2^{\complement} \cap Recv_3) \cup (Recv_1^{\complement} \cap Recv_2 \cap Recv_3)
\]
In the evaluation, we change the number of leaked objects by setting a different percentage when generating the leaked dataset.
At the same time, we fix the total numbers of data objects and receivers in the sharing.

\vspace{3mm}
\noindent\textbf{Evaluation Settings}
We let the value of $|Data|$ be 20,000, and the total number of receivers $N$ be 6.
In the collusive leakage scenario, we evaluate the detection when there are 2, 3, and 6 colluding receivers out of total 6 receivers, and in each case, we randomly try three different types of dataset combinations, which consist of one or more intersection, complement, and union operations.

\vspace{3mm}
\noindent\textbf{Plot Notation}
The $x$-axis indicates the fraction of the leaked data objects over a leaker's dataset or a processed dataset made by colluding receivers.
Specifically, we formally describe the meaning of $x$ using $x = 0.01$ as an example.
\begin{itemize} [leftmargin=*]
	\item In a non-collusive leakage, it means 1\% of guilty receiver's objects $Recv_i$ are leaked and used in \name to identify the leakers.
	\item In the case of a guilty sender, it means 1\% of the common objects are leaked and used in \name to identify the leakers.
	\item As for a collusive leakage, $x = 0.01$ means 1\% objects from the combined dataset (e.g., after a set of collusive operations) are leaked and used in \name.
\end{itemize}
Notice that with the same fraction, the size of a leaked dataset in collusive leakage can be much smaller than that in a non-collusive leakage, e.g., $|Recv_1| \times 0.01 < |Recv_1 \cap Recv_2| \times 0.01$.
The $y$-axis indicates the error rate of the detection results. 
Error rate 1 means the number of leaked objects is not sufficient for the identification algorithm to identify the leaker(s) while error rate 0 indicates that the algorithm can identify all the leakers undoubtedly.
In the plots, we let the maximum accuracy precision to 99.999999\%, in other words, the minimum error rate is $1 \times 10^{-8}$.

\vspace{3mm}
\noindent\textbf{Evaluation Results}
As shown in Figure~\ref{fig:errorrate-amount}, the accuracy becomes 1 when 6\%, 9\%, and 10\% of the received objects are leaked in the case of a non-colluding guilty receiver, a guilty sender, and a number of guilty receivers, respectively.
The more data is leaked and used in \name, the higher the identification accuracy will be.
In Figure~\ref{fig:single-receiver} and Figure~\ref{fig:sender}, we can observe similar identification accuracy in both cases.
This is because $|SenderOnly| = |Unique_i|$ for any $i < N$ and the difference is because of the definition of leakage ratio: a receiver's ratio denominator is $|Recv_i|$ while the sender's ratio denominator is the quantity of all data objects.
Figure~\ref{fig:collusion}, illustrates that when the number of colluding receivers increases, there is a decrease in the detection accuracy.
This is because when more receivers contribute to the leakage, the combined dataset usually becomes smaller in size, e.g., $Recv_1 \cap Recv_2$ has a bigger size than $Recv_1 \cap Recv_2 \cap Recv_3$, and thus a larger number of leaked objects are needed to identify leakers. 
Moreover, we find the accuracy will not change much when different types of combination operations are made in a collusive leakage.

\section{Discussion}
\label{sec:discussion}

\subsection{Basic Trade-off in \name}

The main approach of \name is to allocate data objects among different receivers, which inevitably leads to a fundamental trade-off between \name's performance and the data loss in the allocation.
For example, if each receiver receives the same whole dataset, there is no differentiation that can be used to identify leakers; on the other hand, if each receiver receives a totally different subset, then identification becomes much easier.
Therefore, the larger the data loss is, the more differentiation each receiver's data objects will carry and consequently the more accurate the identification will be.
However, data loss will greatly limit the practical use of a sharing system.
For example, in our example use case, a school is supposed to share all students' information to receivers because the service should serve all the students.

In order to circumvent the limitation, \name introduces fake objects.
This changes the basic trade-off to be \emph{fake data object size versus the accuracy of leaker identification}.
This is because the data loss in normal allocation based approaches is replaced by fake data loss in \name.
On the other hand, the high accuracy is at the cost of having fake objects, which at least double the size of each received dataset.
By enlarging the number of fake data objects used in \name, identification accuracy can be further improved.

\subsection{Impact of Fake Data Objects}
\label{sec:discussion:fake}
PII data is usually shared for two reasons. 
Sometimes it is shared to enable the data receiver to provide certain services for each individual user in the dataset.
For example, online educational services provide personalized services to students based on student information shared by the school.
Other times it is shared to let data receivers analyze the data and dig out useful statistical conclusions.

\name does not bring up negative effects for the former because \name ensures that each receiver will receive all data entries in the original dataset, so all users will be served just as when \name is not applied. 
We acknowledge that a data receiver may have to prepare services for non-existing users (corresponding to the fake data entries that are indistinguishable to the data receiver), but this will not influence the service quality for real users.
Regarding the second use case, we argue that the dataset can be anonymized before data analysis.
In addition, previous works~\cite{modelfordataleakage, data-leakage-detection} also show that properly adding fake data entries basically will not decrease the quality of data analysis.

\begin{figure*} [ht]
	\centering
	\begin{subfigure}{.32\textwidth}
		\centering
		\includegraphics[width=\linewidth]{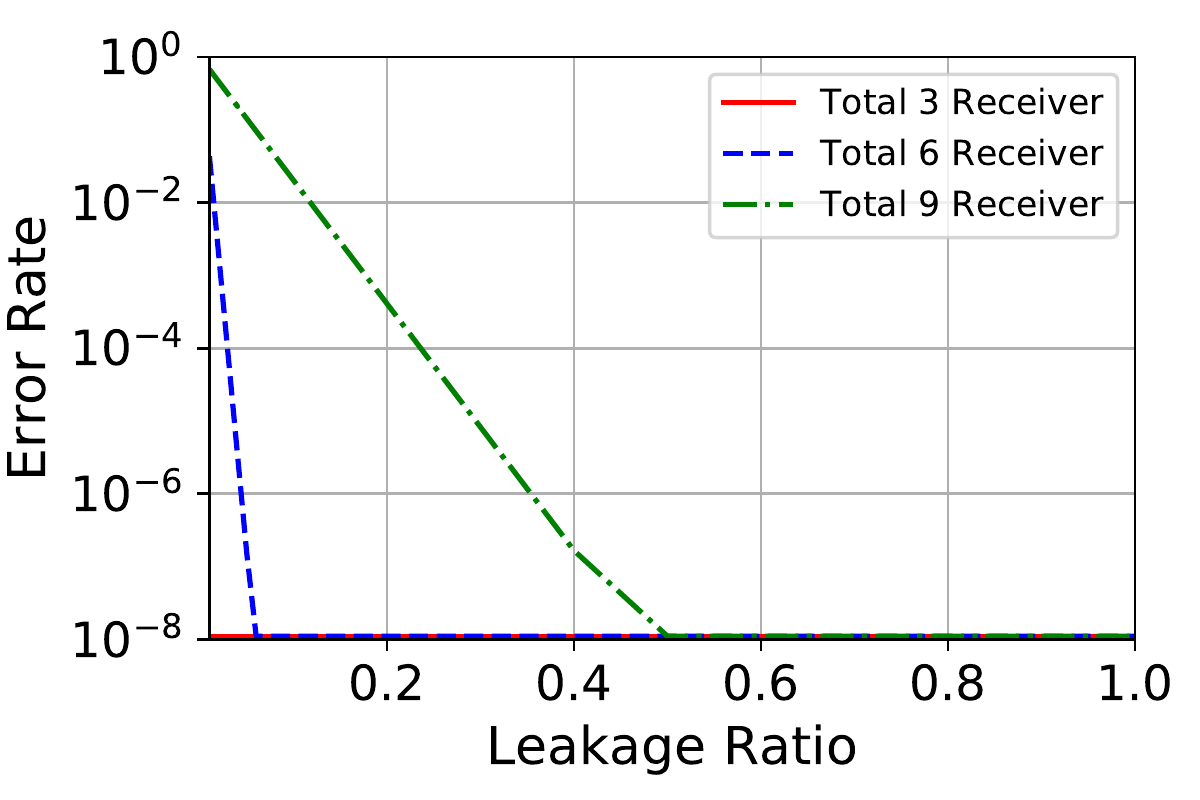}
		\caption{Single Guilty Receiver}
		\label{fig:receiver-num-sender}
	\end{subfigure} %
	~ 
	\begin{subfigure}{.32\textwidth}
		\centering
		\includegraphics[width=\linewidth]{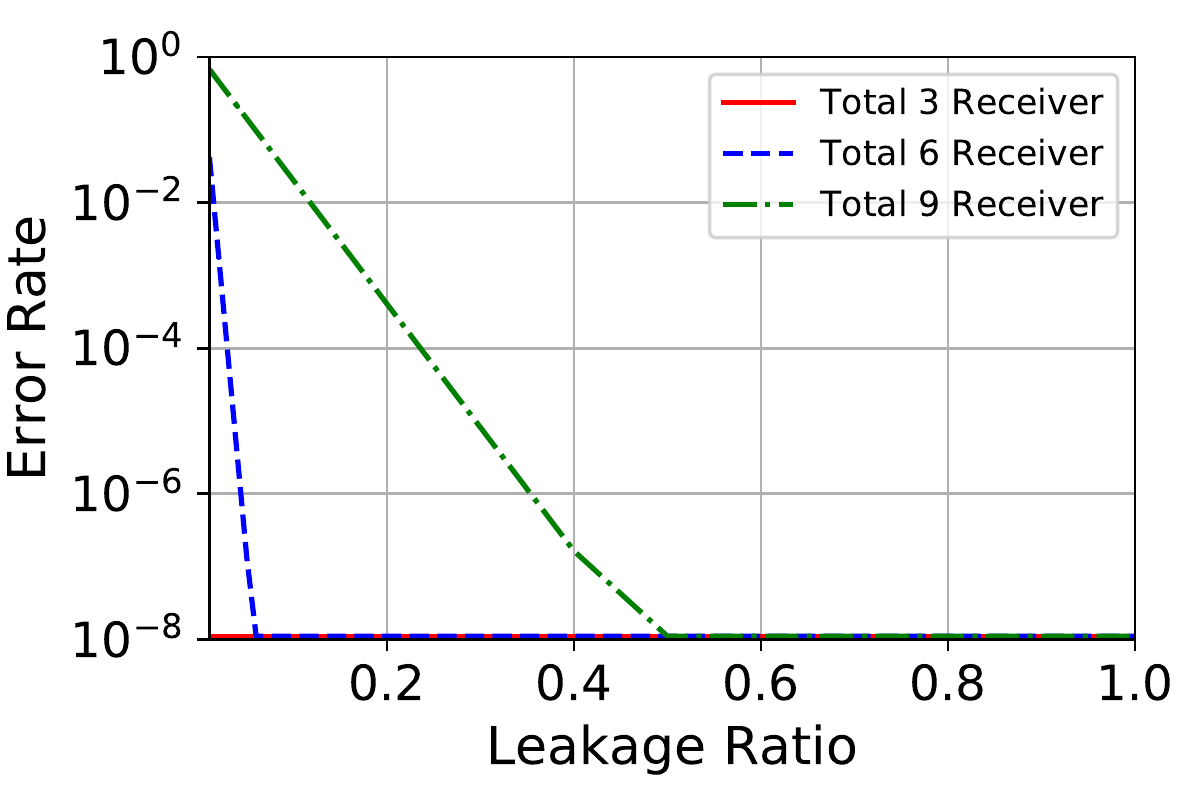}
		\caption{Guilty Sender}
		\label{fig:receiver-num-single}
	\end{subfigure}
	~ 
	\begin{subfigure}{.32\textwidth}
		\centering
		\includegraphics[width=\linewidth]{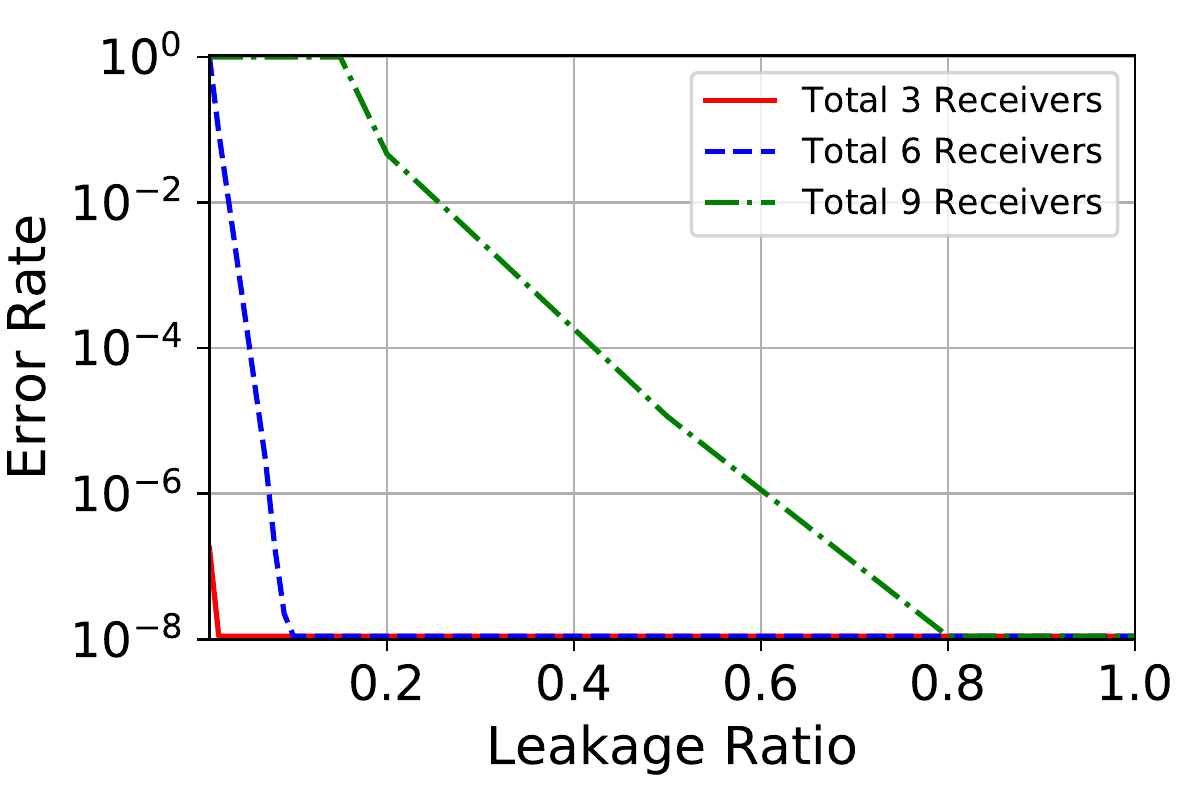}
		\caption{Colluding Receivers}
		\label{fig:receiver-num-collusion}
	\end{subfigure}
	\caption{The Impact of the Total Receiver Number}
	\label{fig:receiver-num}
\end{figure*}

\subsection{\name with a Large Number of Receivers}
\label{sec:discussion:fake-size}

\name allocates fake data objects among receivers for leaker identification.
To ensure a high identification accuracy, adequate fake objects should be allocated to each receiver.
Inadequate objects will degrade \name's performance in leaker identification.
Using the simulation setting as in Section \ref{sec:evaluation}, we set the total fake object number to 20,000.
As shown in Figure~\ref{fig:receiver-num}, when there are totally 3, 6, and 9 receivers, the error rate  increases as the total receiver number grows.
With 3 receivers, \name can identify leakers accurately when about 5\% of data objects are leaked (as shown in Section \ref{sec:evaluation:identification}), so 20,000 fake objects are sufficient.
When sharing with 6 receivers, \name's accuracy is close to 1 when 10\% of data objects are leaked.
This means $20,000$ fake objects can still lead to a relatively good accuracy but is worse relative to the 3-receiver scenario.
However, this quantity of fake objects are inadequate when there are 9 receivers.
As shown, to reach an error ratio less than $10^{-8}$, \name requires 50\%, 50\%, and 80\% of fake data objects leaked from guilty receivers in three leakage cases.

Therefore, in real-world use of \name, the sender should have more fake data objects when sharing data with more receivers.
To achieve an accuracy $>$ 99.99\%, it requires about $200,000$ fake objects in a 9 receiver scenario.

\subsection{Collusion Between a Sender and Receivers}
\label{sec:discussion:sender}

We have not elaborated the case when a sender colludes receivers, but in fact, \name can cope with it naturally.
Specifically, there are two possible situations.
\begin{enumerate} [leftmargin=*]
	\item A sender contributes to the leakage.
	In this case, the only contribution the sender can make are those objects not received by its associates (i.e., colluding receivers).
	By Lemma~\ref{lemma:sender}, there is a great chance that leaked objects overlap $SenderOnly$, so the sender will be caught by Algorithm~\ref{alg:single-leaker}.
	At the same time, all corrupted receivers who contributed to the leakage will be captured by Algorithm~\ref{alg:colluded-leaker}.

	\item The sender does not contribute to the leakage, under which, there is no need to capture the sender.
	Leakers who contributed to the leakage can be captured by Algorithm~\ref{alg:colluded-leaker}.
\end{enumerate}

Importantly,  in either case, the sender has no clue of data received by honest receivers due to the use of OT, so a collusion between the sender and corrupted receivers cannot frame an honest receiver.
\section{Related Work}
\label{sec:related}

The data leakage problem has been explored for years~\cite{shabtai2012survey, cox2002digital}.
Past works in this domain cover leakage prevention, leakage detection, and leaking source identification.
In terms of the media of information, existing works have covered structured data (e.g., rational database, spreadsheets), textual data (e.g., contracts, reports), and multimedia (e.g., images, video).

Our work focuses on the leaker identification problem in PII dataset, which is typically under the category of structured data.
Existing leaker identification solutions to the same issue can be divided into two big categories:
\begin{enumerate*} [label=(\roman*)]
	\item data watermarking,
	and 
	\item data object allocation across receivers.
\end{enumerate*}

\subsection{Watermarking}
In the first category~\cite{halder2010watermarking}, the data sender alters the noise-tolerant attributes or embeds hidden information into the data before sending it to each receiver so that each distributed copy is watermarked.
By associating the watermark contained in the leaked data with the receivers, the data sender is able to identify the guilty party.
This scheme was first proposed by Agrawal et al.~\cite{agrawal2002watermarking, agrawal2003watermarking} and was further explored in later works~\cite{sion2005rights, zhang2006relational, 06watermark-outsourced-db, qin2006watermark, zhou2007additive, halder2010persistent, 2016lime}.

However, as pointed out in many existing works~\cite{agrawal2002watermarking, halder2010watermarking, 2016lime}, it is hard to effectively make a watermark on structured data, because it is trivial for attackers to further alter the structured data to spoil the watermark compared with other types of media (e.g., video, image).
For example, an attacker may undermine the watermark by simply removing or adding noise to certain attributes in the dataset.
Even worse, a group of colluding parties can compare their received copies and figure out the watermarked data or the hidden information.
Consequently, these will lead to unexpected detection results, thus making it hard to be the evidence used in a court of law.

To this end, \name does not follow the watermarking direction; instead, \name takes the data allocation approach and utilizes the property of PII data -- PII attributes of a dataset should be kept unmodified in order to keep the property of PII.
For example, if leakers randomize or erase PII attributes, it becomes harder (i.e., requires much more work) to link leaked information to each individual and the leaked data itself is not PII data any more.
Compared with watermarking approaches, \name has a stronger resistance to data manipulation made by leakers.

\subsection{Data Object Allocation}

First proposed in Panagiotis and Hector's work~\cite{modelfordataleakage, data-leakage-detection}, the second category allocates data objects among receivers so that each receiver will obtain a different data object collection from others.
When leaked data is found, the sender can utilize the receiver-objects mapping to evaluate each receiver's probability of data leaking.
Fake but realistic data objects can also be created to facilitate detection.
Following this direction, Ishu and Ashutosh proposed a guilt detection mechanism~\cite{bigraph} that uses a directed bigraph to allocate data objects to receivers.

However, existing allocation mechanisms also have several issues when used in the real world.
\begin{itemize} [leftmargin=*]
	\item An allocation can be disrupted by a collusion among receivers, by which corrupted receivers can learn the allocation and escape being captured by leaking common objects.
	Furthermore, they can strip special data objects that are potentially used to uniquely identify each receiver.
	
	\item Existing allocation approaches didn't consider the scenario where a data sender itself may leak the data intentionally or unintentionally, which hurts the usability and reliability of existing systems.
	For example, the sender can frame a receiver because it knows this receiver's data and thus can leak those data on purpose.
	In turn, a guilty receiver may also frame the sender when a data breach happens.
	
	\item Furthermore, rare existing systems provide immutable proof of sharing, and thus involved parties may deny what has been shared when data is leaked.
\end{itemize}

In \name, we address these three issues by the use of oblivious transfer, knowledge-based identification algorithm, and undeniable proof of sharing.
\section{Conclusion}
\label{sec:conclusion}
In this work, we study the leaker identification problem in PII data sharing and propose \name, a data-sharing system with reliable leaker identification.
Unlike existing solutions, \name allocates data objects obliviously across receivers, immutably records the sharing process, and accurately identifies leakers with \name's knowledge-based algorithms.
In \name, we also utilize fake data objects to compensate the data loss caused by the allocation through 1-out-of-2 OT, thus ensuring a lossless data sharing.

Not limited to PII data sharing, \name can apply to sharing of other data that contains identifiable attributes, e.g., movie name, product ID, etc.

Our work shows that proactively building auditability into the data sharing can greatly facilitate leaker identification in case of a data leakage.
This will in turn prevent malicious parties from illegally selling or disclosing sensitive information.

\bibliographystyle{IEEEtranS}
\bibliography{wsbf.bib}

\end{document}